\documentclass[11pt,a4paper,english]{amsart}

\usepackage[english]{babel}
\usepackage[T1]{fontenc}
\usepackage{amsmath}
\usepackage{amssymb}
\usepackage{amsthm}
\usepackage{latexsym}
\usepackage{amsfonts}

\usepackage{setspace} 
\setdisplayskipstretch{.1} 
\newlength{\abovebis} 
\newlength{\belowbis} 
\newlength{\aboveshortbis} 
\newlength{\belowshortbis} 
\everydisplay\expandafter{%
  \the\everydisplay 
  \advance\abovedisplayskip\abovebis 
  \advance\belowdisplayskip\belowbis 
  \advance\abovedisplayshortskip\aboveshortbis 
  \advance\belowdisplayshortskip\belowshortbis 
} 

\def\real{\mathbb{R}}

\def\nat{\mathbb{N}}

\def\compl{\mathbb{C}}

\def\proj{\mathbb{C}P}
\def\harm{\mathcal{H}}
\def\meh{\; \!}

\theoremstyle{plain}

\newtheorem{lem}{Lemma}[section]
\newtheorem{theo}[lem]{Theorem}
\newtheorem{prop}[lem]{Proposition}

\newtheorem{cor}[lem]{Corollary}

\theoremstyle{definition}
\newtheorem{defn}{Definition}[section]
\newtheorem{rem}{Remark}[section]

\newenvironment{pfpropbel}{\begin{proof}[\bf Proof of Proposition \ref{propBelAn}.]}{\end{proof}}
\newenvironment{pfmain}{\begin{proof}[\bf Proof of Theorem \ref{maintheo}.]}{\end{proof}}
\newenvironment{pfcoro}{\begin{proof}[\bf Proof of Corollary \ref{coroll}.]}{\end{proof}}

\numberwithin{equation}{section}

\begin{document}
\title[Gel'fand-Calder\'on's inverse problem on bordered surfaces]{Gel'fand-Calder\'on's inverse problem for anisotropic conductivities on bordered surfaces in $\real^3$}
\author{Gennadi Henkin}
\address[G. Henkin]{Universit\'{e} Pierre et Marie Curie\\
case 247, 4, place Jussieu, 75252 Paris Cedex 05}
\email{henkin@math.jussieu.fr}
\author{Matteo Santacesaria}
\address[M. Santacesaria]{Centre de Math\'ematiques Appliqu\'ees, \'Ecole Polytechnique, 91128, Palaiseau, France}
\email{santacesaria@cmap.polytechnique.fr}

\begin{abstract}
Let $X$ be a smooth bordered surface in $\real^3$ with smooth boundary and $\hat \sigma$ a smooth anisotropic conductivity on $X$. If the genus of $X$ is given, then starting from the Dirichlet-to-Neumann operator $\Lambda_{\hat \sigma}$ on $\partial X$, we give an explicit procedure to find a unique Riemann surface $Y$ (up to a biholomorphism), an isotropic conductivity $\sigma$ on $Y$ and a quasiconformal diffeomorphism $F: X \to Y$ which transforms $\hat \sigma$ into $\sigma$.

As a corollary we obtain the following uniqueness result: if $\sigma_1, \sigma_2$ are two smooth anisotropic conductivities on $X$ with $\Lambda_{\sigma_1}= \Lambda_{\sigma_2}$, then there exists a smooth diffeomorphism $\Phi: \overline X \to \overline X$ such that $\Phi|_{\partial X} = \mathrm{Id}$ and $\Phi_{\ast} \sigma_1 = \sigma_2$.
\end{abstract}

\maketitle

\section{Introduction}

Let $X$ be a bordered, oriented, two-dimensional manifold in $\real^3$. We suppose that $X$ possesses a conductivity $\sigma$: this means that we have the following relation
\begin{equation}
j(x)= \sigma(x) \meh du (x), \; x \in X, \label{ohm}
\end{equation}
where $u(x)$ is the voltage potential at $x$, $du (x)$ its differential, and $j(x)$ is the current flowing through $x$. Equation (\ref{ohm}) is just a differential version of Ohm's law. As $j$ is a 1-form, $\sigma$ represents a mapping from 1-forms to 1-forms, i.e. $\sigma$ is a global section of the vector bundle $T(X)^{\ast} \otimes T(X)$ (where $T(X)$, $T(X)^{\ast}$ are respectively the tangent and the cotangent bundle of $X$). It is customary to assume that $\sigma (x)$ is both positive definite and symmetric, in a sense that will be explained later.

We shall also assume that there is no displacement current; thus for any smooth subdomain $X' \subset X$ we have Green's theorem
$$ 0 = \int_{\partial X'} j = \int_{ X'} d j.$$
Since $X'$ is arbitrary, we conclude that $dj = d \sigma \meh d u = 0 $ in $X$.

In general, conductivities are anisotropic; we say that a conductivity is isotropic if the relationship between voltage and current is independent of the direction.

In order to introduce the problem, we define the Dirichlet-to-Neumann operator $\Lambda_{\sigma}: C^1(\partial X) \to L^p ( T(\partial X)^{\ast})$, $p < \infty$ as 
\begin{equation}
\Lambda_{\sigma} f = \sigma \meh du|_{\partial X},
\end{equation}
where $\sigma \in C^3(T(X)^{\ast} \otimes T(X))$, $f \in C^1(\partial X)$ and $u$ is the unique $W^{1,p}(\overline X)$-solution of the Dirichlet problem
\begin{equation} \label{DP}
d \sigma \meh d u =0 \; \textrm{on } X, \; u|_{\partial X} = f.
\end{equation}
Our aim is to answer the following question, that is a variation of an inverse boundary value problem posed by Gel'fand \cite{Ge} and Calder\'on \cite{C}: which information about $X$ and $\sigma$ can be extracted from the mapping $\Lambda_{\sigma}$?

The main result of this paper is:

\begin{theo}  \label{maintheo}
Let $X$ be a bordered, $C^3$, oriented, two-dimensional manifold with $C^3$ boundary and let $\hat \sigma$ be a $C^3$-anisotropic conductivity on $X$. From the Dirichlet-to-Neumann operator $\Lambda_{\hat \sigma} : C^1(\partial X) \to L^p ( T(\partial X)^{\ast})$, $p < \infty$, and from the knowledge of the genus of $X$, we can find by an explicit procedure:
\begin{itemize}
\item[\textit{i)}] a bordered Riemann surface $Y$,

\item[\textit{ii)}] an isotropic conductivity $\sigma$ on $Y$,

\item[\textit{iii)}] a $C^3$ diffeomorphism $F : X \to Y$ such that $F_{\ast} \hat \sigma = \sigma$.
\end{itemize}
Moreover, if $\tilde Y$ is another Riemann surface, $\tilde \sigma$ an isotropic conductivity on $\tilde Y$ and $\tilde F: X \to \tilde Y$ a $C^3$ diffeomorphism such that $\tilde F_{\ast} \hat \sigma = \tilde \sigma$, then $\Psi = \tilde F \circ F^{-1} : Y \to \tilde Y$ is a biholomorphism such that $\Psi_{\ast} \sigma = \tilde \sigma$.
\end{theo}
The push-forward of a conductivity $\sigma$ by a diffeomorphism $\Phi:\overline X \to \overline Y$ is defined, following \cite[\S 1]{Sylvester}, as
\begin{equation} \label{push}
(\Phi_{\ast} \sigma) \alpha = \Phi_{\ast} (\sigma (\Phi^{\ast} \alpha)),
\end{equation}
where $\Phi^{\ast} \alpha$ denotes the pull-back of the 1-form $\alpha$ and $\Phi_{\ast}=(\Phi^{-1})^{\ast}$ denotes the pull-back by $\Phi^{-1}$ acting on the 1-form $\sigma (\Phi^{\ast} \alpha)$.

L. Tartar was the first to remark (see \cite{Kohn}) that, when $\Phi:\overline X \to \overline X$, this new conductivity $\Phi_{\ast} \sigma$ has the same boundary measurements as $\sigma$ if $\Phi|_{\partial X}=\mathrm{Id}$, where $\mathrm{Id}$ is the identity map. Thus, it is clearly not possible to determine $\sigma$ uniquely from $\Lambda_{\sigma}$; more specifically we cannot find more than \textit{i)--iii)} from $\Lambda_{\sigma}$. This is pointed out in the following corollary of our main result.

\begin{cor} \label{coroll}
Let $X$ be a bordered, $C^3$, oriented, two-dimensional manifold with $C^3$ boundary and let $\sigma_1, \sigma_2$ be two $C^3$-anisotropic conductivities on $X$. If $\Lambda_{\sigma_1} = \Lambda_{\sigma_2}$ then there exists a $C^3$ diffeomorphism $\Phi : \overline X \to \overline X$ such that $\Phi|_{\partial X} = \mathrm{Id}$ and $\sigma_2 = \Phi_{\ast} \sigma_1$.
\end{cor}

\subsection*{Historical remarks} These results were obtained earlier only for $X \subset \real^2$. Even for this case, Theorem \ref{maintheo} was obtained only recently by the authors \cite{HS}, using arguments and results taken from \cite{Novikov}, \cite{Nachman}, \cite{Sylvester}, \cite{Astala}, \cite{Gutarts}.

Corollary \ref{coroll}, for $X \subset \real^2$, was proved in an original paper by Sylvester \cite{Sylvester} for $C^3$ conductivities close to constants (the last restriction was eliminated in \cite{Nachman}).
From \cite{Sylvester} one can deduce (see \cite{HM}) that for any bordered surface $X$, equipped with an anisotropic conductivity, there exists a unique complex structure, i.e. $d = \overline \partial + \partial$, for which the equation $d \hat \sigma \meh d u =0$ transforms into $d \sigma \meh d^c u=0$, where $\sigma$ is a positive function (which represents an isotropic conductivity) and $d^c = i (\overline \partial - \partial)$.

In the context of surfaces Corollary \ref{coroll} is the first uniqueness result for the inverse anisotropic-conductivity problem. A uniqueness result (even with partial data) for the inverse isotropic-conductivity problem was recently obtained in \cite{Gui}, using stationary phase techniques from \cite{Bu}.

From the reconstruction viewpoint, Theorem \ref{maintheo} is the first result on the recovering of the above-mentioned complex structure of a bordered surface with known genus from its Dirichlet-to-Neumann operator. A method for recovering isotropic conductivities on surfaces with known genus was recently developed in  \cite{HN}. In addition, a reconstruction procedure for complex surfaces with constant conductivity was obtained in \cite{HM}.

\subsection*{Scheme of the proofs}
The main idea behind this paper is the same as in \cite{Sylvester}, i.e. to reduce the problem to the isotropic case.

We first equip our real surface with some complex structure (e.g. the complex structure induced by the euclidean metric of $\real^3$) and then we embed the surface in the complex affine space $\compl^3$ as a domain $X$ on a nonsingular affine algebraic curve $V$. Next, we extend $\hat \sigma$ by a constant on $V \setminus X$. Successively, we find a global analogue of isothermal coordinates, uniquely determined on $V$ by a given anisotropic conductivity and natural asymptotic conditions. This is accomplished by proving existence and uniqueness of special solutions of a certain Beltrami equation; here we follow the works started by Gauss \cite{Ga} and fully developed by Ahlfors \cite{Ahlfors} and Vekua \cite{Vekua}, along with the Hodge-Riemann decomposition \cite{Ho} and the generalization of related operator estimates.

We cannot expect, like in the plane, that the deformed surface will live in the same compactified surface after the change of coordinates. Thus, we will find a new surface $W$ where our conductivity is isotropic; in general this surface will be algebraic, but possibly with intersection points.

Thanks to this global Beltrami solution $F$ and results in \cite{HN} for the isotropic case, we can prove existence and uniqueness of Faddeev-type aniso\-tro\-pic functions $\hat \psi_{\theta} (z,\lambda)$ on $V$: a two-parameter family of solutions of the anisotropic conductivity equation \eqref{DP} on $V$ with exponential asymptotics (see \eqref{faddeev}), originally introduced in \cite{Faddeev}. We will also prove a formula, inspired by \cite{Sylvester}, that allows us to reconstruct the boundary values $F|_{\partial X}$ of our Beltrami solution starting from $\hat \psi|_{\partial X}$. We then show how to reconstruct $\hat \psi|_{\partial X}$ from the knowledge of $\Lambda_{\hat \sigma}$, through a Fredholm-type integral equation.

The reconstruction procedure works as follows: starting from $\Lambda_{\hat \sigma}$ one reconstructs $\hat \psi|_{\partial X}$ and then $F|_{\partial X}$; thus one recovers $\Gamma = F(\partial X)$ and also $\Lambda_{F_{\ast} \hat \sigma} = \Lambda_{\sigma}$. Since $\Gamma$ has to be the boundary $\partial Y$ of a Riemann surface $Y$, one recovers that surface through Cauchy-type formulas. Finally, from the knowledge of $\Lambda_{\sigma}$, the application of results in \cite{HN} yields $F_{\ast} \hat \sigma = \sigma$ on $Y$.

Our scheme can be summarized in the following diagram:

$$\Lambda_{\hat \sigma} \to \hat \psi|_{\partial X} \to F|_{\partial X} \to \partial Y \to Y \to \sigma.$$

\subsection*{An open problem} It is known that, for constant conductivities, the Di\-rich\-let-to-Neu\-mann operator for $dd^c u=0$ determines the genus of a surface; this is a consequence of results in \cite{LU}, \cite{Bel}, \cite{HM} and \cite{GG}. These results can be generalized to the case of conductivities close to constants.

In the general case of non-constant conductivities, the unique determination of the genus of a bordered surface from its Dirichlet-to-Neumann operator is still an open question.

\section{Preliminaries} \label{prelim}

\subsection{Basic definitions} \label{basic}

Let us provide more details about the objects discussed in the introduction. 

We say that a conductivity $\sigma$ is positive definite and symmetric, if, for $a,b \in T_x(X)^{\ast}$, $x \in X$,
\begin{align}
a \wedge \sigma b &= b \wedge \sigma  a, \label{sigma1} \\
a \wedge \sigma a &= \varphi (x) dx^1 \wedge dx^2, \; \; \varphi (x) \geq C_{\sigma} |a|^2 > 0, \label{sigma2}
\end{align}
where $x^1, x^2$ are positively oriented coordinates and $| \meh \meh|$ is the euclidean norm. From (\ref{sigma1}) and (\ref{sigma2}) one sees that locally, in the chart $(U_{\alpha}, x_{\alpha})$, our conductivity can be written as
\begin{equation}
\sigma|_{U_{\alpha}} = \sum_{i,j =1}^2 \sigma_{\alpha}^{ij} (-1)^{j-1} dx_{\alpha}^{3-j} \wedge \frac{\partial}{\partial x_{\alpha}^i},
\end{equation}
where the matrix $(\sigma_{\alpha}^{ij})$ is positive definite and symmetric ($>C_{\sigma} I$).

With this notation, an isotropic conductivity $\sigma$ is just a conductivity whose associated matrix has the form $(\sigma^{ij})=\sigma_0 I$, where $\sigma_0 : X \to \real$ is a bounded positive function and $I$ is the identity matrix.

Equation (\ref{DP}) now reads locally
\begin{equation}
d \sigma \meh d u = \left( \sum_{i,j=1}^2 \frac{\partial}{\partial x^i} \left( \sigma^{ij} \frac{\partial u}{\partial x^j} \right) \right) dx^1 \wedge dx^2 = 0. \label{diri}
\end{equation}

Let us now explain some general properties of the push-forward of a conductivity.
Let $\Phi : \overline X \to \overline Y$ be a diffeomorphism between bordered surfaces and $\sigma$ a conductivity on $X$. We define the push-forward $\Phi_{\ast} \sigma$ of $\sigma$ as in (\ref{push}); locally, it reads
\begin{equation} \nonumber
\Phi_{\ast} \sigma =  \left( \frac{{}^t (D\Phi) \sigma (D\Phi)}{|\det(D\Phi)|}\right) \circ \Phi^{-1},
\end{equation}
where $D\Phi$ is the matrix differential of $\Phi$ and $\sigma$ is seen as its associated matrix.

We recall that if $\Phi$ is conformal, then $\frac{{}^t (D\Phi)}{|\det(D\Phi)|} = (D\Phi)^{-1}$, thus the push-forward of an isotropic conductivity by a conformal diffeomorphism is still isotropic.

We would also like to compare the two Dirichlet-to-Neumann operators $\Lambda_{\sigma}$ and $\Lambda_{\Phi_{\ast} \sigma}$. By pull-back properties, if $u$ satisfies $d \sigma \meh d u = 0$, then $\Phi_{\ast} u = u \circ \Phi^{-1}$ satisfies $d \Phi_{\ast} \sigma \meh d (\Phi_{\ast} u) =0$. This fact implies that the unique solution of the Dirichlet problem
\begin{equation} \label{DP2}
d (\Phi_{\ast} \sigma) \meh d v =0 \; \textrm{on } Y, \; v|_{\partial Y} = f \circ (\Phi|_{\partial X})^{-1}
\end{equation}
is just $v = \Phi_{\ast} u$, where $u$ is the unique solution of
\begin{equation}
d \sigma \meh d u =0 \; \textrm{on } X, \; u|_{\partial X} = f .
\end{equation}
So if $Y = X$ and $\Phi_{\partial X} = \textrm{Id}$ we see that $\Lambda_{\Phi_{\ast} \sigma} = \Lambda_{\sigma}$; in general, it is important to underline the fact that $\Lambda_{\Phi_{\ast} \sigma}$ is completely determined by $\Lambda_{\sigma}$ and $\Phi|_{\partial X}$.

\subsection{Complex viewpoint}

Here we will introduce some complex notation. We define standard complex coordinates $z = x^1+ix^2$, $\overline z= x^1 - i x^2$, $dz =d x^1 + i dx^2$, $d \overline z = dx^1 - i dx^2$, $\frac{\partial}{\partial z} = \frac 1 2 \left( \frac{\partial}{\partial x^1} -i \frac{\partial}{\partial x^2} \right)$, $\frac{\partial}{\partial \overline z} = \frac 1 2 \left( \frac{\partial}{\partial x^1} +i \frac{\partial}{\partial x^2} \right)$.

We can now rewrite the conductivity $\sigma$ with the complex coordinates; we obtain
\begin{equation}
\sigma|_{U_{\alpha}} = (\sigma_{\alpha}^0(-i dz_{\alpha}) + \overline \sigma_{\alpha}^1 (i d \overline z_{\alpha})) \wedge \frac{\partial}{\partial z_{\alpha}} + (\sigma_{\alpha}^1 (-i dz_{\alpha}) + \sigma_{\alpha}^0 (i d \overline z_{\alpha})) \wedge \frac{\partial}{\partial \overline z_{\alpha}}
\end{equation}
where
\begin{equation}
\sigma^0 = \frac{\sigma^{22}+\sigma^{11}}{2}, \; \; \; \sigma^1= \frac{\sigma^{11} - \sigma^{22}}{2} - i \sigma^{12}.
\end{equation}
We have chosen to represent the image of $\sigma$ in the basis $\{-i \meh dz, \; i d \overline z\}$ in order to have the hermitian matrix $\left( \begin{matrix} \sigma^0 & \sigma^1 \\ \overline \sigma^1 & \sigma^0 \end{matrix} \right).$

One verifies that these new coefficients satisfy the following transformation rules
\begin{equation} \label{transf}
\sigma_{\alpha}^0 = \sigma_{\beta}^0, \; \textrm{and} \; \sigma_{\alpha}^1=\sigma_{\beta}^1 \frac{d z_{\beta}}{d z_{\alpha}} \overline{\left( \frac{d z_{\beta}}{d z_{\alpha}} \right)}^{-1}
\end{equation}

Let us remark that, if $\sigma$ is isotropic, represented by the matrix $\sigma_0 I$, then equation \eqref{diri} reads
$$d \sigma \meh d u = d \sigma_0 \meh d^c u = 0.$$
Throughout all the paper we will always identify an isotropic conductivity $\sigma$ with its associated function $\sigma_0$ to simplify notation; thus the conductivity equation, in this case, will always be written $d \sigma \meh d^c u = 0$ and $\Lambda_{\sigma} f = \sigma d^c u|_{\partial X}$, with $u$ the solution of \eqref{DP}.

\subsection{Embedding in projective space} \label{secembedding}
Let $\proj^3$ be the complex projective space with homogeneous coordinates $w= (w_0 : w_1 : w_2 : w_3)$ and let $\proj^2_{\infty} = \{ w \in \proj^3 \; : \; w_0=0 \}$. Then $\proj^3 \setminus \proj^2_{\infty}$ can be considered as a complex affine space with coordinates $z_k= w_k / w_0$, $k= 1,2,3$.
Thanks to a classical result of G. Halphen (cfr. \cite[Prop. 6.1]{Ha}) any compact Riemann surface of genus $g$ can be embedded in $\proj^3$ as a projective algebraic curve $\tilde V$, which intersects $\proj^2_{\infty}$ transversally in $d > g$ points, where $d \geq 1$ if $g =0$, $d \geq 3$ if $g=1$ and $d \geq g+3$ if $g \geq 2$.

Without loss of generality we can assume the following facts:
\begin{itemize}
\item[i)] $V= \tilde V \setminus \proj^2_{\infty}$ is a connected affine algebraic curve in $\compl^3$ defined by polynomial equations $V = \{ z \in \compl^3 : p_1(z)= p_2(z) = p_3(z) =0 \}$ such that $\mathrm{rank} \! \! \left[ \frac{\partial p_1}{\partial z}(z), \frac{\partial p_2}{\partial z}(z), \frac{\partial p_3}{\partial z}(z)\right] \equiv 2, \; \forall z \in V$;

\item[ii)] $\tilde V \cap \proj^2_{\infty} = \{ \beta_1, \ldots , \beta_d \}$, where
$$ \beta_l= (0 : \beta^1_l : \beta^2_l : \beta^3_l), \left( \frac{\beta^2_l}{\beta^1_l}, \frac{\beta^3_l}{\beta^1_l} \right) \in \compl^2, \; l=1, 2, \ldots, d;$$

\item[iii)] for $r_0 >0$ large enough
$$\det \left[ \begin{array}{cc} \frac{\partial p_{\alpha}}{\partial z_2} & \frac{\partial p_{\alpha}}{\partial z_3} \\ \frac{\partial p_{\beta}}{\partial z_2} & \frac{\partial p_{\beta}}{\partial z_3} \end{array} \right] \neq 0, \; \mathrm{for} \; z \in V : |z_1| \geq r_0 \; \mathrm{and} \; \alpha \neq \beta;$$

\item[iv)] for $|z|$ sufficiently large we have
$$\frac{dz_2}{dz_1}|_{V_l}=\gamma_l + O \left( \frac{1}{z^2_1} \right), \; \frac{dz_3}{dz_1}|_{V_l}=\tilde \gamma_l + O \left( \frac{1}{z^2_1} \right),$$
where $\gamma_l, \tilde \gamma_l \neq 0$, for $l= 1, \ldots,d, \; d \geq 2$, $V_0 = \{ z \in V  :  |z_1| \leq r_0 \}$ and $V \setminus V_0 = \cup_{l=1}^d V_l$ (the $V_l$ are the connected components of $V \setminus V_0$, for $l=1, \ldots, d$).
\end{itemize}
We equip $\tilde V$ with the projective volume form $d d^c \log (1+|z|^2)$ and $V$ with the euclidean volume form $d d^c |z|^2$; we can thus consider the spaces $L^p_{0,1}(\tilde V)$ and $L^p_{0,1}(V)$ of $L^p$ (0,1)-forms, equipped with the norms $\| \; \|_{L^p_{0,1}(\tilde V)}$ and $\| \; \|_{L^p_{0,1}(V)}$, respectively. There is a canonical surjective map $C^{\infty}_{0,1}(\tilde V) \to C^{\infty}_{0,1}(V)$, so that we can compare the two above-defined norms; indeed, for $p \geq 2$ and $f \in L^p_{0,1}(\tilde V)$, we have that $\| f\|_{L^p_{0,1}(V)} \leq \| f \|_{L^p_{0,1}(\tilde V)}$ (in particular $\| f\|_{L^2_{0,1}(V)} = \| f \|_{L^2_{0,1}(\tilde V)}$). This yields the inclusion $L^p_{0,1}(\tilde V) \subset L^p_{0,1}(V)$, for $p \geq 2$ (through the canonical map), and the same result is true for (1,0)-forms, i.e. $L^p_{1,0}(\tilde V) \subset L^p_{1,0}(V)$, for $p \geq 2$.

In section \ref{secBeltrami}, the norm $\| \; \|_p$ will always stand for the affine norm $\| \; \|_{L^p_{0,1}(V)}$ (or $\| \; \|_{L^p_{1,0}(V)}$), although it will be use to make some estimates on forms defined on the whole compact surface $\tilde V$.\smallskip

We now define the spaces $\tilde W^{1,p} (\tilde V) = \{ F \in L^{\infty} (\tilde V) : \partial F \in L^{ p}_{1,0} (\tilde V) \}$, $\tilde W_{0,1}^{1,p} (\tilde V) = \{ F \in L_{0,1}^{\infty} (\tilde V) : \partial F \in L^{ p}_{1,1} (\tilde V) \}$ for $1< p < \infty$ and $H_{0,1}(\tilde V)$ the space of antiholomorphic (0,1)-forms on $\tilde V$. 

From the Hodge-Riemann decomposition theorem we have, for every $\Phi_0 \in W^{1,p}_{0,1}(\tilde V)$, $\Phi_0 = \overline \partial ( \overline \partial^{\ast} G \Phi_0) + \harm\Phi_0$, where $\harm\Phi_0 \in H_{0,1} (\tilde V)$ is defined as
$$\harm \Phi_0= \sum_{j=1}^g \left(\int_V \Phi_0 \wedge \omega_j \right) \overline \omega_j,$$
with $\{ \omega_j \}$ an orthonormal basis of holomorphic (1,0)-forms on $\tilde V$ and $G$ is the Hodge-Green operator for the Laplacian $\overline \partial \overline \partial^{\ast} + \overline \partial^{\ast} \overline \partial$ on $\tilde V$ with the following properties: $G(H_{0,1} (\tilde V))=0$, $\overline \partial G= G\overline \partial$, $\overline \partial^{\ast} G = G \overline \partial^{\ast}$.

We also define the operator $R$, for $f \in C^{\infty}_{0,1}(\tilde V)$, as $Rf (x)  = \overline \partial^{\ast} Gf (x) - \overline \partial^{\ast} G f (\beta_1)$; we will see, as a consequence of Lemma \ref{lemest}, that $R : L^p_{0,1}(\tilde V) \to \tilde W^{1,p}(\tilde V)$, for $2 < p < \infty$.

In the rest of the paper we will suppose for simplicity that $V= \{ z \in \compl^2 : P(z) = 0 \}$ is an affine algebraic curve in $\compl^2$. 

\subsection{Remarks on the extension of $\hat \sigma$ on $V \setminus X$}

In the following of the paper, we will always suppose that $\hat \sigma$ is the identity in a neighbourhood of $\partial X$ (i.e. its associated matrix is the identity). In this way we could easily extend $\hat \sigma$ to $V$ by putting $(\hat \sigma^{ij}) = I$ on $V \setminus X$, and this new conductivity will still be $C^3$.

This simplification is possible thanks to the following construction.
After embedding $X = X_1$ as an open set of the affine algebraic curve $V \subset \compl^2$ above, we can find an open set $X_2 \subset V$ with the following properties:
\begin{itemize}
\item[i)] $X_1 \subset X_2 \subset V$,
\item[ii)] $X_2$ has a $C^1$ boundary (the same smoothness as $\partial X$),
\item[iii)] $\hat \sigma$ can be extended to $X_2$ as a $C^3$ conductivity $\hat \sigma'$ such that $\hat \sigma' \equiv I$ in a neighbourhood of $\partial X_2$.
\end{itemize}
This is possible because one can reconstruct $\hat \sigma|_{\partial X_1}$ and its derivatives at the boundary from $\Lambda_{\hat \sigma}$ as in \cite{Kohn2}.

Thus we only have to show that $\Lambda_{\hat \sigma'}$ can be determined by $\Lambda_{\hat \sigma}$ and $\hat \sigma'|_{X_2 \setminus X_1}$. This can be done as in \cite[Sec. 6]{Nachman}.

The Dirichlet-to-Neumann maps $\Lambda^{ij}$ are defined as follows: we consider, for $i, j=1,2$, $f_j \in C^{1}(\partial X_j)$ and $u_j \in C^1(X_2 \setminus \overline X_1)$ the solution of the Dirichlet problem $d \hat \sigma' \meh d u_i = 0$ in $X_2 \setminus \overline X_1$ such that $u_1|_{\partial X_1} = f_1$, $u_1|_{\partial X_2} = 0$, respectively $u_2|_{\partial X_1} = 0$, $u_2|_{\partial X_2} = f_2$. Then we define
$$\Lambda^{ij} f_j = \hat \sigma' du_j|_{\partial X_i}$$
and we have the following relation.

\begin{prop}
Under our assumption, $\Lambda_{\hat \sigma} - \Lambda^{11}$ is an invertible operator $\Lambda_{\hat \sigma} - \Lambda^{11} : C^{1}(\partial X_1) \to C^{0}(\partial X_1)$ and
\begin{equation} \nonumber
\Lambda_{\hat \sigma'} = \Lambda^{22} + \Lambda^{21} (\Lambda_{\hat \sigma} - \Lambda^{11})^{-1} \Lambda^{12}.
\end{equation}
\end{prop}
The proof of this formula follows from the definition of the operators. The fact that $\Lambda_{\hat \sigma} - \Lambda^{11}$ is invertible comes from an explicit construction of its inverse, which turns out to be the single-layer operator on $\partial X_1$ corresponding to the Green function $G$ for the Dirichlet problem on $X_2$. More explicitly, it is the operator
$$S f (x) = \int_{\partial X_1} G(x,y) f(y) dy,$$
where $G$ satisfies $d \hat \sigma' \meh d G= -\delta(x-y)$ in $X_2$ and $G(\cdot,0)|_{\partial X_2}=0$. 

\section{The Beltrami Equation} \label{secBeltrami}
In this section we will study the equation

\begin{equation} \label{beltrami}
\overline \partial w = \mu \partial w,
\end{equation}
called the Beltrami equation, on a Riemann surface. Here $\mu$ is a bounded (-1,1)-form, namely a Beltrami differential, whose definition we will recall.
\begin{defn}
A Beltrami differential $\mu(z) \frac{\overline{dz}}{dz}$ on a Riemann surface $V$, equipped with an atlas $\{ U_{\alpha} , z_{\alpha} \}$, is a collection of $L^{\infty}$ complex-valued function $\mu^{\alpha}$ defined on $z_{\alpha} (U_{\alpha})$ such that
\begin{equation}
\mu^{\alpha} (z_\alpha) = \mu^{\beta} (z_{\beta}) \frac{\overline{ \left( \frac{dz_{\beta}}{dz_{\alpha}} \right)}  }{\frac{dz_{\beta}}{dz_{\alpha}}}
\end{equation}
and $\|\mu \|_{\infty} = \sup_{\alpha} \| \mu^{\alpha} \|_{\infty} < 1$.
\end{defn}
With this definition, equation (\ref{beltrami}) is valid globally.

The main result of this section is:
\begin{theo} \label{teobeltrami}
Let $X \subset V$ be an open subset of an affine Riemann surface $V$, let $\tilde V \supset V$ be its compactification, as in section \ref{prelim}, and let $\mu$ be a Beltrami differential on $\tilde V$ with $\mathrm{supp}(\mu) \subset X$ and $\| \mu \|_{\infty} \leq k < 1$. Then, for $j=1,2$, there is a unique solution $w_j(z)$ of equation (\ref{beltrami}) on $V$ such that $w_j(z)= z_j + w_{0j} (z)$, $w_{0j} \in \tilde W^{1,p}(\tilde V)$ for $p > 2$ and $w_{0j} (\beta_1) = 0$.
\end{theo}

In order to prove this theorem we introduce the operator $\Pi = \partial R$, initially defined on smooth forms, and we show some estimates which slightly generalize a result by Calder\'on and Zygmund; these will yield in particular that $\Pi : L^p_{0,1}(\tilde V) \to L^p_{1,0}(\tilde V)$, for $2 \leq p < \infty$.

We recall (see section \ref{secembedding} for further explanations) that the norm $\| \; \|_p$ stand for the affine norm $\| \; \|_{L^p_{0,1}(V)}$ (or $\| \; \|_{L^p_{1,0}(V)}$). 

\begin{lem} \label{lemest}
For $f \in L^2_{0,1}(\tilde V) \cap \ker(\harm)$ we have
\begin{equation}\label{est1}
\| \Pi f\|_2 = \|f\|_2
\end{equation}
and, for $f \in L^p_{0,1}(\tilde V) \cap \ker(\harm)$, $p >2$
\begin{equation}\label{est2}
\| \Pi f\|_p \leq C_p \|f\|_p ,  \; \mathrm{and}  \lim_{p \to 2^+} C_p =1.
\end{equation}
\end{lem}
\begin{proof}
The proof is given for $f \in C^2_{0,1}(\tilde V) \cap \ker(\harm)$; the original statement will follow by a density argument. We have the following chain of equalities, where by Stokes' theorem and the Hodge decomposition on $\tilde V$
\begin{align*}
\|\Pi f\|_2^2 &= \int_{V} \partial R f \wedge \overline{\partial Rf} = - \int_{V} R f \wedge \partial \overline{\partial Rf} = - \int_{V} R f \wedge \partial \overline R \overline{\partial f} \\
&= - \int_{V} R f \wedge  \overline{\partial f}= \int_{V} \overline \partial R f \wedge \overline f = \int_{V} f \wedge \overline{f} = \|f\|_2^2.
\end{align*}

To prove (\ref{est2}) we first decompose the operator $\Pi$ in the following way
\begin{equation} \label{decompo}
\Pi f = \Pi_1 f + \Pi_2 f  = \int_{|\zeta -z| \leq \delta} f(\zeta) \Pi_1(\zeta, z) + \int_{|\zeta - z| > \delta} f(\zeta) \Pi_2 (\zeta, z),
\end{equation}
for $\delta$ sufficiently small, where in affine coordinate form
\begin{equation}
\Pi_1(\zeta, z) = \frac{d \zeta \wedge dz}{2 \pi i(\zeta -z)^2} (1+ \varepsilon(\delta)), \; \varepsilon(\delta) \to 0 \; \textrm{when} \; \delta \to 0
\end{equation}
and $\Pi_2$ is bounded. Decomposition \eqref{decompo} gives a so-called \textit{parametrix} for the operator $\Pi$.
From the Calder\'on-Zygmund result for the operator $$F \mapsto \lim_{\varepsilon \to 0} \frac{1}{2\pi i} \int_{|\zeta - z| > \varepsilon} \frac{F(z)}{(z-\zeta)^2} dz d \overline z$$ (see \cite[p.106]{Ahlfors}) where $f = F d \overline z$ , we have the estimate
$\| \Pi_1 f \|_p \leq (1 + \varepsilon(\delta)) \tilde C_p \|f \|_p$. In addition, we also have $\| \Pi_2 f \|_p \leq \| \Pi_2 \|_{L^{\infty}(|\zeta - z| > \delta)} \|f \|_p = K(\delta) \|f \|_p.$ Putting it all together we find that
\begin{equation}
\| \Pi f\|_p \leq ( (1 + \varepsilon(\delta) ) \tilde C_p + K(\delta)) \|f\|_p = C_p \|f\|_p.
\end{equation}
The fact that $C_p \to 1$ when $p \to 2$ is a consequence of the Riesz-Thorin interpolation theorem (see \cite[Thm. 1.1.1, p.2]{Bergh}) and of (\ref{est1}).
\end{proof}

Now, using the last lemma, we fix $p > 2$ such that $kC_p < 1$. The proof of the theorem will be given for the case $j=1$; the other case is completely analogous.
\begin{proof}[Proof of Theorem \ref{teobeltrami}.]
Let us begin with the existence statement. We look for solutions of the form $w(z) = z_1 + Rf$.
Thus
\begin{align*}
\overline \partial w &= \overline \partial R f= f - \harm f, \\
\partial w &= dz_1 + \partial Rf = dz_1 + \Pi f = dz_1 + \Pi (f -\harm f),
\end{align*}
for $R \harm f = 0$ (and so $\Pi \harm f = 0$). If we impose equation \eqref{beltrami}, we obtain an integral equation for $f_0 = f- \harm f$:
\begin{equation}\label{equa1}
f_0-\mu \Pi f_0 = \mu dz_1.
\end{equation}
Under our assumptions, the linear operator $f \mapsto \mu \Pi f$ is a contraction in $L^p_{0,1} (\tilde V) \cap \ker (\harm)$ (its norm is $ \leq k C_p <1$), so the series
$$f_0 = \mu dz_1 + \mu \Pi \mu dz_1 + \mu \Pi \mu \Pi \mu dz_1 + \ldots$$
converges in $L^p_{0,1}(\tilde V) \cap \ker (\harm)$ to a solution of (\ref{equa1}). Then we define $w_{01} = R f_0$ which satisfies $\partial w_{01} = \Pi f_0 \in L^p_{1,0}(\tilde V)$ and $w_{01} \in L^{\infty}(\tilde V)$ (the latter follows from properties of $R$). Thus the function $w(z) = z_1 + w_{01}(z)$ is a solution of (\ref{beltrami}).

To show uniqueness, we first remark that $w_{01} = R \overline \partial w_{01}$. This follows from the fact that $\overline \partial w_{01} = \overline \partial w = \mu \partial w = \mu (dz_1 + \partial w_{01}) \in L^p_{0,1}(\tilde V)$ because the support of $\mu$ is contained in $X$; we can thus calculate $R\overline \partial w_{01}$ and see that $\overline \partial (w_{01} - R\overline \partial w_{01}) = 0$. Now $w_{01} - R \overline \partial w_{01}$ is a bounded holomorphic function which goes to zero for $z \to \beta_1$, so it vanishes. In particular, this yields $w = z_1 + R \overline \partial w$.

Now, if $w'= z_1 + w'_{01} = z_1 + R \overline \partial w'$ is another solution, we obtain
$$\partial (w - w') = \Pi \mu (\partial (w - w')),$$
which gives $\partial (w -w')=0$ thanks to our estimates, and also $\overline \partial (w -w') =0$ because of the Beltrami equation. So $w-w'$ must be constant, and in fact it vanishes because of our normalisation.
\end{proof}

\subsection{Properties of the solution} \label{normalization}

We now consider the application $F: V \to \compl^2$ defined as $F(z)=(w_1(z), w_2(z))$ where $w_1, \; w_2$ are the solutions of the Beltrami equation given by Theorem \ref{teobeltrami}. In particular, we want to understand the image surface $W = F(V)$.

By \cite[Thm. 2, p.97]{Ahlfors} one has that $F$ is a local homeomorphism; besides, since $w_1$ and $w_2$ are solutions of the Beltrami equation, $W$ has a holomorphic atlas. Thus, by classical results, it is an algebraic curve as well, but possibly with intersection points. Let us note that, by the properties of $F$, we have $W \cap \proj^1_{\infty} = V \cap \proj^1_{\infty}$.

\subsection{Applications to anisotropic conductivities}

The most important consequence of Theorem \ref{teobeltrami}, for this paper, is the following proposition about the existence of global isothermal coordinates which transforms an anisotropic conductivity into an isotropic one.

\begin{prop} \label{propBelAn}
Let $X \subset V$ be an open subset of an affine Riemann surface $V$, let $\tilde V \supset V$ be its compactification, as in section \ref{prelim}, and $\hat \sigma$ a $C^k$-anisotropic conductivity on $V$ ($k \geq 1$), represented by the identity matrix on $V \setminus X$. Then there exists a unique affine algebraic curve $W$, and a unique $C^k$ immersion $F: V \to W$, $F=(w_1,w_2)$ such that $F_{\ast} \hat \sigma = \sigma$ is isotropic on $W$ (where $F^{-1}$ exists) and $w_j(z)= z_j + w_{0j}(z)$ with $w_{0j} \in \tilde W^{1,p}(\tilde V)$ and $w_{0j}(\beta_1) = 0$, for $j=1,2$.
\end{prop}

We will need the following lemma:

\begin{lem} \label{lemBelAn}
Let $X \subset V \subset \tilde V$ as in proposition \ref{propBelAn}. Then every conductivity $\sigma$ on $X$, extended on $V \setminus X$ by the identity matrix, determines a Beltrami differential $\mu_{\sigma} \frac{\overline{dz}}{dz}$ with support contained in $X$ given locally by
\begin{equation}
\mu_{\sigma}^{\alpha}=\mu_{\sigma}|_{U_{\alpha}} = \frac{\sigma_{\alpha}^{22} - \sigma_{\alpha}^{11} - 2i \sigma_{\alpha}^{12}}{\sigma_{\alpha}^{11} + \sigma_{\alpha}^{22} +2 \sqrt{\det(\sigma_{\alpha})}} = \frac{- \overline \sigma_{\alpha}^1}{\sigma_{\alpha}^0 + \sqrt{(\sigma_{\alpha}^0)^2 - | \sigma_{\alpha}^1|^2}}.
\end{equation}
\end{lem}

\begin{proof}
From the transformation rules (\ref{transf}) one immediately has the relation
$$\mu_{\sigma}^{\beta} = \frac{- \overline \sigma_{\alpha}^1}{\sigma_{\alpha}^0 + \sqrt{(\sigma_{\alpha}^0)^2 - | \sigma_{\alpha}^1|^2}} \frac{\overline{ \left( \frac{dz_{\alpha}}{dz_{\beta}} \right)}  }{\frac{dz_{\alpha}}{dz_{\beta}}} = \mu_{\sigma}^{\alpha} \frac{\overline{ \left( \frac{dz_{\alpha}}{dz_{\beta}} \right)}  }{\frac{dz_{\alpha}}{dz_{\beta}}}.$$
In addition, we have that 
\begin{equation} 
|\mu_{\sigma}|^2 = \frac{\sigma^{11} + \sigma^{22} - 2 \sqrt{\det \sigma}}{\sigma^{11} + \sigma^{22} + 2 \sqrt{\det \sigma}} \leq k < 1
\end{equation}
and $\mu_{\sigma} \equiv 0$ outside $X$.
\end{proof}

\begin{pfpropbel}
We define $\mu_{\hat \sigma}$ as in lemma \ref{lemBelAn}; by Theorem \ref{teobeltrami} we can construct $F (z) = (w_1 (z) , w_2 (z))$, $F : \compl^2 \to \compl^2$, where $w_1, w_2$ are the special solutions of the Beltrami equation $\overline \partial w_j = \mu_{\hat \sigma} \partial w_j$. Using \cite[Prop. 1.3]{Sylvester}, we have that $F_{\ast} \hat \sigma = \sigma I$ is isotropic on $F(V) = W \subset \tilde W$ but defined only where $F^{-1}$ is. In particular we have
\begin{equation}
(F_{\ast} \hat \sigma) (w) = (\det \hat \sigma)^{1/2} \circ F^{-1}(w) I = \sigma(w) I,
\end{equation}
where $\sigma(w)= \sqrt{\sigma_0^2(z(w))-|\sigma_1(z(w))|^2}$.

By remarks of section \ref{normalization} we have that $F$ is an immersion: it is a $C^k$ immersion because of smoothness assumptions on $\hat \sigma$. 
%
%
\end{pfpropbel}

\section{Faddeev-type Anisotropic Solutions}

In this section we generalise the results of \cite{HN}, by proving existence and uniqueness of a family of special solutions of the anisotropic conductivity equation, so-called Faddeev-type solutions.

Let us recall from \cite{HN} the definitions of a few operators. We equip $V$ with the Euclidean volume form $dd^c |z|^2$, and let $\varphi \in L^1_{1,1}(V) \cap L^{\infty}_{1,1}(V),\; f \in \tilde W^{1,p}_{1,0}(V) = \{ F \in L_{1,0}^{\infty} (V) : \overline \partial F \in L^p_{1,1}(V) \}$, for $p > 2,\; \lambda \in \compl \setminus \{ 0 \}$ and $\theta \in \compl$. We define
\begin{align*}
\hat R_{\theta} \varphi &= R ((dz_1 + \theta dz_2) \rfloor \varphi) \wedge (dz_1+ \theta dz_2), \\
R_{\lambda, \theta} f &= e_{-\lambda,\theta} \overline{R (\overline{e_{\lambda,\theta}f})}, \; \textrm{where} \; \; e_{\lambda, \theta}(z)=e^{\lambda(z_1+\theta z_2)-\overline \lambda (\overline z_1 + \overline \theta \overline z_2)}.
\end{align*}
Let $\hat \sigma$ be a $C^3$ anisotropic conductivity on $V$ with $\hat \sigma \equiv I$ on $V \setminus X$ and $\hat a_1 , \ldots, \hat a_g \in V \setminus X$ an effective divisor.

\begin{defn}
A function $\hat \psi_{\theta} (z,\lambda)$, with $\theta, \lambda \in \compl, \; z \in V$, is called a Faddeev-type function on $V$ associated with $\hat \sigma, \theta, \lambda$ and $\{ \hat a_1, \ldots, \hat a_g \} \subset V \setminus X$, if
\begin{equation} \label{faddeev}
d \hat \sigma \meh d \hat \psi_{\theta} (z, \lambda) = 2 \left( \sum_{j=1}^g \hat C_{j,\theta}(\lambda) \delta(z, \hat a_j) \right) e^{\lambda(z_1 + \theta z_2)}, \; z \in V,
\end{equation}
and $\hat \psi_{\theta} (z, \lambda) e^{-\lambda(z_1 + \theta z_2)} \to \hat K_l$ (constant), when $z \in V_l$, $z \to \infty$, for $l=1, \ldots, d$ with the normalisation $\hat K_1=1$.
\end{defn}
Let $F: V \to W$ be the mapping constructed in Proposition \ref{propBelAn}, $Y = F(X)$, $a_j = F (\hat a_j)$ for $j=1,\ldots,g$ and $\sigma = F_{\ast} \hat \sigma$ the isotropic conductivity on $W$. Let $\psi_{\theta}(w, \lambda)$ be the Faddeev-type isotropic functions on $W$ constructed in \cite{HN} as the solutions of 
\begin{equation} \label{isofaddeev}
d\sigma d^c \psi_{\theta} (w,\lambda) = 2\left( \sum_{j=1}^g C_{j,\theta}(\lambda) \delta(w, a_j) \right) e^{\lambda(w_1 + \theta w_2)}
\end{equation}
with $\psi_{\theta} e^{-\lambda(w_1+ \theta w_2)} \to K_l$ (constants, with $K_1=1$), when $w \in W_l$, $w \to \infty$, for $l=1, \ldots, d$, where $W_l=F(V_l)$.

We also define 
\begin{equation} \label{delta}
\Delta_{\theta} (\lambda) = \det \left[ \int_{\eta \in W} \hat R_{\theta}(\delta (\eta, a_j)) \wedge \overline \omega_k (\eta) e_{\lambda, \theta}(\eta) \right]_{j,k=1,\ldots, g}
\end{equation}
where $\{ \omega_k \}$ is an orthonormal basis of holomorphic (1,0)-forms on $\tilde W$, and we call $E_{\theta} = \{ \lambda \in \compl : \Delta_{\theta} (\lambda) = 0\}$.

\begin{theo} \label{teofaddeev}
For any generic $\theta$, $\{ \hat a_1, \ldots, \hat a_g \}$ and $\lambda \in \compl \setminus E_{\theta}, |\lambda| \geq const(V, \;\{\hat a_j\},\; \theta, \; \hat \sigma)$ there exists a unique Faddeev-type solution $\hat \psi_{\theta} (z, \lambda)$ associated with $\hat \sigma, \; \theta,\; \lambda$ and $\{\hat a_1, \ldots, \hat a_g \}$. Moreover $E_{\theta}$ is a closed, nowhere dense subset of $\compl$ and we have the equality
\begin{equation} \label{identity}
\hat \psi_{\theta} (z, \lambda) = \psi_{\theta} (F(z), \lambda), \; z \in V
\end{equation}
\end{theo}

\begin{proof} We will here provide a complete proof of Theorem \ref{teofaddeev} when the Beltrami solution $F$, given by proposition \ref{propBelAn}, is an embedding; at the end we will indicate necessary corrections for the proof of the general case.

With this assumption, by proposition \ref{propBelAn} there exists a unique diffeomorphism $F(z) = (w_1(z) , w_2(z))$ such that $w_j (z)= z_j + w_{0j}(z)$, $w_{0j} \in \tilde W^{1,p} (\tilde V), \; p >2$ and $F_{\ast} \hat \sigma = \sigma$ is isotropic on the image.

By \cite[Prop. 1.1]{HN}, the set $E_{\theta}$ is closed and nowhere dense and by \cite[Thm. 1.1]{HN}, for every $\lambda \in \compl \setminus E_{\theta}, |\lambda| \geq const(W, \{a_j \}, \theta, \sigma)$ there exists a unique Faddeev-type isotropic function $\psi_{\theta} (w,\lambda)$ as defined in \eqref{isofaddeev}.

Now let $\hat \psi_{\theta} (z,\lambda)$ be an anisotropic Faddeev-type solution. We consider $\psi'_{\theta}(w,\lambda) = \hat \psi_{\theta} (F^{-1}(w),\lambda)$ and see that
$$d\sigma d^c \psi'_{\theta} (w,\lambda) =2 \left( \sum_{j=1}^g C'_{j,\theta}(\lambda) \delta(z, a_j) \right) e^{\lambda(F_1^{-1}(w) + \theta F_2^{-1}(w))}$$
from the construction of $\sigma$ and the definition of $a_j$. Using the properties of $F$ (in particular that $F \to \mathrm{Id}$ for $z \to \infty$) and of $\hat \psi_{\theta}$, we have that $\psi'_{\theta} e^{-\lambda(w_1+ \theta w_2)} \to K_l$ with $K_1=1$; this shows that $\psi'_{\theta}$ and $\psi_{\theta}$ satisfy the same asymptotic conditions. Thus, by the uniqueness of $\psi_{\theta} (w,\lambda)$ we obtain the identity \eqref{identity}, which proves both existence and uniqueness for the case where $F$ is a diffeomorphism.

If $F$ is just an immersion the result is still valid; we can follow the same outline of the proof, taking into account the following:
\begin{itemize}
\item[i)] in the definition \eqref{delta} we have to use weakly holomorphic forms $\omega_k$, i.e. forms such that $\omega_k \in H^{1,0}(W \setminus Sing\, W)$ and $\omega_k$ are bounded on $W$ in a neighbourhood of $Sing \, W$;
\item[ii)] we say that $u$ is a solution of $d \sigma d^c u =0$ on $W \setminus \{a_1, \ldots, a_g\}$, for $a_1, \ldots, a_g \in Reg \, W$ if $u$ is locally bounded on $W \setminus \{a_1, \ldots, a_g\}$ and $d \sigma d^c u =0$ on $Reg \, W \setminus \{a_1, \ldots, a_g\}$,
\item[iii)] Proposition 1.1 and Theorem 1.1 of \cite{HN} are still valid for $W$ with points of simple self-intersection, but in the proofs one has to make some minor modifications in order to make estimates for operators $\hat R$ and $R_{\lambda}$.
\end{itemize}
The properties i) and ii) show that the holomorphic forms $\omega_k$ and the functions $u$ can be smoothly extended to a normalization of $W$.
\end{proof}

We now prove a formula, motivated by \cite[Prop. 2.7]{Sylvester}, which will play a key role in the reconstruction procedure. 

\begin{theo}
Let $\hat \psi_{\theta}$ be the Faddeev-type anisotropic functions constructed above. Then for every $z \in V \setminus X$ (in particular for $z \in \partial X$), for every $ \varepsilon >0$ and generic $\theta \in \compl$ we have 
\begin{equation} \label{recbel}
\lim_{\lambda \to \infty} \inf_{ \{ \lambda' : | \lambda' - \lambda | \leq \varepsilon \} } \frac{\log \hat \psi_{\theta}(z,\lambda')}{\lambda'} = w_1(z) + \theta w_2 (z) 
\end{equation}
\end{theo}
\begin{proof}
We will use the following essential property of $\Delta_{\theta} (\lambda)$ from \cite[Prop. 1.1]{HN}, i.e., for every $\varepsilon > 0$
\begin{equation} \label{propdelta}
\underline \lim_{\lambda \to \infty} \sup_{ \{ \lambda' : | \lambda' - \lambda | \leq \varepsilon \} } |\Delta_{\theta} (\lambda')| |\lambda|^g > 0.
\end{equation}
Using \cite[Prop. 3.1]{HN} and \eqref{propdelta}, for $z \in V \setminus X$ we have $\sigma (F (z)) =1$, $$\hat \psi_{\theta}(z,\lambda) = e^{\lambda (F_1(z) + \theta F_2 (z))} \mu_{\theta} (F(z), \lambda),$$ 
$$\inf_{ \{ \lambda' : | \lambda' - \lambda | \leq \varepsilon \}} |\mu_{\theta}(\lambda') -1|= O \left( \frac{1}{\lambda^{1-0}}\right), \; \lambda \to \infty.$$
Thus one obtains
\begin{align*}
&\inf_{ \{ \lambda' : |\lambda' - \lambda| \leq \varepsilon \} } \frac{\log \hat \psi_{\theta} (z,\lambda')}{\lambda'} = w_1(z) + \theta w_2(z) + \inf_{ \{ \lambda' : |\lambda' - \lambda| \leq \varepsilon \} } \frac{\log \mu_{\theta} (w(z), \lambda')}{\lambda'} \\ 
&= w_1(z) + \theta w_2(z) + O\left( \frac{\log \lambda}{\lambda} \right) \to w_1(z) + \theta w_2(z), \; \textrm{as} \; \lambda \to \infty. \qedhere
\end{align*}
\end{proof}

\section{An Integral Equation for $\hat \psi_{\theta}|_{\partial X}$}

In this section we show how one can reconstruct the boundary values $\hat \psi_{\theta}|_{\partial X}$ from the Dirichlet-to-Neumann operator through a Fredholm-type integral equation.

Following the approach of Gutarts \cite{Gutarts}, based on Eskin \cite[Thm. 18.5]{E}, we decompose the differential operator $d \hat \sigma \meh d$ as $dd^c - Q$, where $Q$ is a compactly supported operator. Faddeev-type anisotropic functions, $\hat \psi_{\theta}(z,\lambda)=e^{\lambda(z_1+\theta z_2)}\hat \mu_{\theta}(z,\lambda)$, then satisfy

\begin{align} \label{eqGut}
&dd^c \hat \psi_{\theta}(z,\lambda) = Q \hat \psi_{\theta}(z,\lambda) + 2 \left( \sum_{j=1}^g \hat C_{j,\theta}(\lambda) \delta(z, \hat a_j) \right) e^{\lambda(z_1 + \theta z_2)},\\ \label{equamu}
&\overline \partial \left(\partial + \lambda (dz_1+ \theta dz_2) \right) \hat \mu_{\theta} = \frac{i}{2}Q \hat \mu_{\theta} + i \sum_{j=1}^g \hat C_{j,\theta}(\lambda) \delta(z, \hat a_j). 
\end{align}

\begin{theo} \label{theo51}
We have
\begin{itemize}
\item[\textit{i)}] For every $\lambda \in \compl \setminus E_{\theta}$, $|\lambda| \geq const(V, \{ a_j \}, \theta, \hat \sigma)$ the boundary values of $\hat \psi_{\theta}$ satisfy the following integral equation:
\begin{align}\label{fred} 
&\hat \psi_{\theta} (z,\lambda)|_{\partial X} = \frac{i}{2}\int_{\zeta \in \partial X} \! \! \! \! G_{\lambda,\theta}(z,\zeta) (\Lambda_{\hat \sigma} - \Lambda_0) \hat \psi_{\theta}(\zeta, \lambda) \\ \nonumber
&\qquad + i e^{\lambda(z_1 + \theta z_2)} \sum_{j=1}^g \hat C_{j,\theta}(\lambda) g_{\lambda,\theta}(z, \hat a_j) \\ \nonumber
&\qquad - \lim_{\varepsilon \to 0} \frac i 2  \! \! \!  \! \! \! \!  \! \! \! \! \! \int\limits_{ \{\zeta \in V : |\zeta-z| \geq \varepsilon\} }\! \! \! \! \! \! \! \! \! \! \hat \psi_{\theta}^0(\zeta,\lambda)dd^c G_{\lambda,\theta}(z,\zeta)  \\ \nonumber
&\qquad -\lim_{R \to \infty} \int\limits_{|\zeta_1| =R}[\bar \partial G_{\lambda,\theta}(z,\zeta) \hat \psi_{\theta}^0(\zeta,\lambda)+ G_{\lambda,\theta}(z,\zeta)\partial \hat \psi_{\theta}^0(\zeta,\lambda)],
\end{align}
with
\begin{equation} \label{54}
i \sum_{j=1}^g (\hat a_{j,1} + \theta \hat a_{j,2})^{-k} \hat C_{j,\theta} (\lambda) = - \int_{z \in \partial X} (z_1 + \theta z_2)^{-k} e^{-\lambda (z_1+ \theta z_2)} \overline{\Lambda_{\hat \sigma} \overline{\hat \psi_{\theta}} (z, \lambda)},
\end{equation}
for $k = 2, \ldots, g+1$,
\begin{equation} \nonumber
G_{\lambda,\theta}(z,\zeta) =e^{\lambda [(z_1 - \zeta_1) + \theta (z_2 -\zeta_2)]} g_{\lambda, \theta} (z, \zeta),
\end{equation}
$g_{\lambda, \theta} (z, \zeta)$ is the kernel of the operator $R_{\lambda, \theta} \circ \hat R_{\theta}$, \\ $\Lambda_0 f = d^c u |_{\partial X}$ where $dd^c u= 0$ on $X$ and $u|_{\partial X} = f$, \\
$\hat \psi_{\theta}^0(\zeta,\lambda)$ is a continuous function for $\zeta \in V \setminus ( \bigcup_j \{ a_j \} )$ such that
\begin{align*}
\hat \psi_{\theta}^0(\cdot,\lambda) |_{V \setminus X} &=\hat \psi_{\theta}(\cdot,\lambda)|_{V \setminus X},\\
dd^c \hat \psi_{\theta}^0 &=0 \text{ on } X.
\end{align*}
\item[\it ii)] Equation (\ref{fred}) is a Fredholm-type integral equation and has a unique solution in \\ $W^{1,2}(\partial X)$, $\forall \lambda \in \compl \setminus E_{\theta}$, $|\lambda| \geq const(V, \{ a_j \}, \theta, \hat \sigma)$.
\end{itemize}
\end{theo}

\begin{rem}
Theorem \ref{theo51} is a generalization of \cite[Thm. 1.2A]{HN} to the anisotropic case. Note that the term $e^{\lambda(z_1+\theta z_2)}$ in the right hand side of the integral equation in \cite[Thm. 1.2A]{HN} must be replaced by the term
\begin{align*}
&- \lim_{\varepsilon \to 0} \frac i 2  \! \! \!  \! \! \! \!  \! \! \! \! \! \int\limits_{ \{\zeta \in V : |\zeta-z| \geq \varepsilon\} }\! \! \! \! \! \! \! \! \! \! \psi_{\theta}^0(\zeta,\lambda)dd^c G_{\lambda,\theta}(z,\zeta)  \\ \nonumber
&-\lim_{R \to \infty} \int\limits_{|\zeta_1| =R}[\bar \partial G_{\lambda,\theta}(z,\zeta)  \psi_{\theta}^0(\zeta,\lambda)+ G_{\lambda,\theta}(z,\zeta)\partial \psi_{\theta}^0(\zeta,\lambda)],
\end{align*}
like in formula \eqref{fred} above. It is important to note that the function $\hat \psi_{\theta}^0$ in \eqref{fred} can be represented using $\hat \psi_{\theta}(\cdot, \lambda)|_{\partial X}$ by Poisson-type formulas on $X$ and $V \setminus X$:
\begin{align} \label{poisson1}
&\hat \psi_{\theta}^0(\zeta,\lambda)= \int\limits_{w \in \partial X} \hat \psi_{\theta}(w,\lambda) \partial g_+^0(\zeta,w), \quad \text{if } \zeta \in X, \\ \label{poisson2}
&\hat \psi_{\theta}^0(\zeta,\lambda)= - \int\limits_{w \in \partial X} \hat \psi_{\theta}(w,\lambda) \partial g_-^0(\zeta,w) - i \sum_{j=1}^g e^{\lambda(a_{j,1}+\theta a_{j,2})} C_{j,\theta}(\lambda)g_-^0(\zeta,a_j),
\end{align}
if $\zeta \in V \setminus X$, where $g_+^0$ is the Green function for the Laplacian $\bar \partial \partial$ on $X$ such that $g_+^0(\cdot,0)|_{\partial X}=0$, and $g_-^o$ is a Green function for $\bar \partial \partial \psi = 0$ on $V \setminus X$ with the condition $g_-^0(\cdot,0)|_{\partial X}=0$ and $\psi(\zeta) = e^{\lambda (\zeta_1 + \theta \zeta_2)}O(1)$, $\zeta \to \infty$. The existence of such a Green function on $V \setminus X$ follows from \cite[Lemma 4.1]{HN}.
\end{rem}

In order to prove Theorem \ref{theo51} we will need the following equality:
\begin{lem} \label{lem52}
For $\lambda \in \compl \setminus E_{\theta}$, $|\lambda| \geq const(V, \{ a_j \}, \theta, \hat \sigma)$ and $z \in V$ we have
\begin{align} \label{equaext}
&e^{\lambda(z_1+\theta z_2)}+ \lim_{\varepsilon \to 0} \frac i 2  \! \! \!  \! \! \! \!  \! \! \! \! \! \int\limits_{ \{\zeta \in V : |\zeta-z| \geq \varepsilon\} }\! \! \! \! \! \! \! \! \! \! \hat \psi_{\theta}(\zeta,\lambda)dd^c G_{\lambda,\theta}(z,\zeta) \\ \nonumber
&\qquad +\lim_{R \to \infty} \int\limits_{|\zeta_1| =R}[\bar \partial G_{\lambda,\theta}(z,\zeta) \hat \psi_{\theta}(\zeta,\lambda)+ G_{\lambda,\theta}(z,\zeta)\partial \hat \psi_{\theta}(\zeta,\lambda)] =0
\end{align}
\end{lem}

\begin{proof}
We write $\hat \mu_{\theta}$ as the solution of the integral equation
\begin{align} \label{intmu}
\hat \mu_{\theta}(z,\lambda) = 1 + \frac{i}{ 2} \int_{\zeta \in X} g_{\lambda,\theta}(z,\zeta) Q \hat \mu_{\theta}(\zeta,\lambda) + i \sum_{j=1}^g \hat C_{j,\theta}(\lambda) g_{\lambda,\theta}(z,a_j),
\end{align}
for $z \in V$. The equivalence between \eqref{equamu} and \eqref{intmu} implies the equality
\begin{equation} \nonumber
\hat \mu_{\theta}(z,\lambda)= 1 + \int_{\zeta \in V} g_{\lambda,\theta}(z,\zeta) \bar \partial (\partial + \lambda (d \zeta_1+\theta d \zeta_2)) \hat \mu_{\theta}(\zeta,\lambda),
\end{equation}
which becomes, using integration by parts,
\begin{align*}
&\hat \mu_{\theta}(z,\lambda)=1 + \int_{\zeta \in V} \bar \partial (\partial - \lambda (d \zeta_1+\theta d \zeta_2)) g_{\lambda,\theta}(z,\zeta)  \hat \mu_{\theta}(\zeta,\lambda) \\
&\qquad + \lim_{R \to \infty}\int\limits_{|\zeta_1|=R} \! \! [\bar \partial g_{\lambda,\theta}(z,\zeta) \hat \mu_{\theta}(\zeta,\lambda)+g_{\lambda,\theta}(z,\zeta)(\partial + \lambda (d \zeta_1+\theta d \zeta_2))\hat \mu_{\theta}(\zeta,\lambda)].
\end{align*}
Now, in order to obtain \eqref{equaext}, it is sufficient to prove the following limit:
\begin{equation} \label{limitgreen}
\lim_{\varepsilon \to 0} \int\limits_{ \{\zeta \in V: |\zeta - z| \leq \varepsilon\} }\bar \partial (\partial - \lambda (d \zeta_1+\theta d \zeta_2)) g_{\lambda,\theta}(z,\zeta)  \hat \mu_{\theta}(\zeta,\lambda) = \hat \mu_{\theta}(z,\lambda).
\end{equation}
This limit is based on the following formula
\begin{align} \label{symgreen}
&G_{\lambda,\theta}(z,\zeta) - \overline{G_{-\lambda,\theta}(\zeta,z)} \\ \nonumber
& = - \int_{w \in V}G_{\lambda,\theta}(w,\zeta)e^{\bar \lambda[(\bar w_1-\bar z_1)+\bar \theta(\bar w_2-\bar z_2)]} \overline{\harm_{-\lambda, \theta}(\hat R(\delta(\cdot,z)))} \wedge \lambda(d w_1 + \theta d w_2) \\ \nonumber
&\quad + \int_{w \in V}\overline{G_{-\lambda,\theta}(w,z)}e^{\lambda[(w_1-\zeta_1)+\theta(w_2-\zeta_2)]} \harm_{\lambda, \theta}(\hat R(\delta(\cdot,\zeta))) \wedge \bar \lambda(d \bar w_1 + \bar \theta d \bar w_2),
\end{align}
where
\begin{align*}
\harm_{\lambda, \theta}(\hat R(\delta(\cdot,\zeta))) &= e_{-\lambda,\theta}\harm(e_{\lambda,\theta}(\hat R(\delta(\cdot,\zeta)))),\\
e_{\lambda,\theta}(w) &= e^{\lambda(w_1+\theta w_2)-\bar \lambda( \bar w_1+ \bar \theta \bar w_2)}.
\end{align*}
The proof of \eqref{symgreen} follows the proof of a classical theorem about the symmetry of the classical Green function (see \cite[p.434]{Gam}), combined with the following statement from \cite[Remark 1.2]{HN}
\begin{equation} \label{green73}
\bar \partial( \partial + \lambda(d z_1 + \theta d z_2))g_{\lambda,\theta}(z,\zeta) = \delta(z,\zeta) + \bar \lambda (d \bar z_1 + \bar \theta d \bar z_2) \wedge \harm_{\lambda,\theta}(\hat R ( \delta(z,\zeta))).
\end{equation}
Limit \eqref{limitgreen} is now given by formula \eqref{symgreen} and the following estimates:
\begin{align*}
\lim_{\varepsilon \to 0} &\int\limits_{ \{ \zeta \in V:|\zeta - z| \leq \varepsilon \} } \bar \lambda (d \bar \zeta_1+ \bar \theta d \bar \zeta_2) \wedge \harm_{\lambda,\theta}(\hat R(\delta(\zeta,z)))\hat \mu_{\theta}(\zeta,\lambda) = 0, \\
&\int_V \bar \lambda (d  \bar \zeta_1+ \bar \theta d \bar \zeta_2) \wedge \harm_{\lambda,\theta}(\hat R(\delta(\zeta,z)))\hat \mu_{\theta}(\zeta,\lambda) < \infty . \qedhere
\end{align*}
\end{proof}

\begin{proof}[Proof of Theorem \ref{theo51}]
\textit{i)} Like in the isotropic case (see \cite[Lemmas 3.1, 3.3]{HN}) a solution $\hat \psi_{\theta}$ of the differential equation \eqref{eqGut} can be characterized as a solution of the integral equation
\begin{align} \label{intGut}
&\hat \psi_{\theta} (z,\lambda) = \frac i 2 \int_{\zeta \in X}G_{\lambda,\theta}(z,\zeta)  Q \hat \psi_{\theta}(\zeta, \lambda) \\ \nonumber
&\qquad + e^{\lambda (z_1 + \theta z_2)}+ i e^{\lambda(z_1 + \theta z_2)} \sum_{j=1}^g \hat C_{j,\theta}(\lambda) g_{\lambda,\theta}(z, \hat a_j),
\end{align}
where $\{\hat C_{j,\theta}(\lambda)\}$ satisfy \eqref{54}. Indeed \eqref{eqGut} implies that $\partial \hat \psi_{\theta}$ is holomorphic on $V \setminus (X \cup \bigcup_j \{ \hat a_j\})$, the estimate
\begin{equation} \nonumber
\partial \hat \psi_{\theta} = e^{z_1 + \theta z_2} O(1), \quad \textrm{for } z \to \infty
\end{equation}
and the equality
\begin{equation} \nonumber
\mathrm{Res}_{\hat a_j} \partial \hat \psi_{\theta} = \frac{\hat C_{j,\theta}}{2 \pi}e^{\hat a_{j,1} + \theta \hat a_{j,2}}.
\end{equation}
The residue theorem applied to the form
\begin{equation} \nonumber
\frac{e^{-(z_1 + \theta z_2)}\partial \hat \psi_{\theta}}{(z_1+\theta z_2)^k}
\end{equation}
gives \eqref{54}.

Now, using equality \eqref{equaext}, for $z \in V \setminus X$ we obtain
\begin{align*}
&\frac{i}{2} \int_{\partial X} e^{\lambda [(z_1 - \zeta_1) + \theta (z_2 - \zeta_2)]} g_{\lambda, \theta} (z, \zeta) (\Lambda_{\hat \sigma} - \Lambda_0) \hat \psi_{\theta}(\zeta, \lambda) \\
&\qquad = \frac{i}{2}\int_{\partial X} G_{\lambda,\theta}(z,\zeta) (\Lambda_{\hat \sigma} - \Lambda_0) \hat \psi_{\theta}(\zeta, \lambda) \\
&\qquad = \frac{i}{2}\int_{X} G_{\lambda,\theta}(z,\zeta) dd^c \hat \psi_{\theta}(\zeta, \lambda)- \frac{i}{2}\int_X dd^c G_{\lambda,\theta}(z,\zeta) \, [\hat \psi_{\theta}(\zeta, \lambda) -\hat \psi_{\theta}^0(\zeta, \lambda)]\\
&\qquad = \frac{i}{2}\int_X G_{\lambda,\theta}(z,\zeta) Q \hat \psi_{\theta} + \lim_{\varepsilon \to 0} \frac i 2  \! \! \!  \! \! \! \!  \! \! \! \! \! \int\limits_{ \{\zeta \in V : |\zeta-z| \geq \varepsilon\} }\! \! \! \! \! \! \! \! \! \! \hat \psi_{\theta}^0(\zeta,\lambda)dd^c G_{\lambda,\theta}(z,\zeta) + e^{\lambda (z_1 + \theta z_2)}\\
& \qquad \qquad  +\lim_{R \to \infty} \int\limits_{|\zeta_1| =R}[\bar \partial G_{\lambda,\theta}(z,\zeta) \hat \psi_{\theta}^0(\zeta,\lambda)+ G_{\lambda,\theta}(z,\zeta)\partial \hat \psi_{\theta}^0(\zeta,\lambda)] \\
&\qquad = \hat \psi_{\theta}(z,\lambda)- i e^{\lambda(z_1 + \theta z_2)} \sum_{j=1}^g \hat C_{j,\theta}(\lambda) g_{\lambda,\theta}(z, \hat a_j)\\
&\qquad \qquad + \lim_{\varepsilon \to 0} \frac i 2  \! \! \!  \! \! \! \!  \! \! \! \! \! \int\limits_{ \{\zeta \in V : |\zeta-z| \geq \varepsilon\} }\! \! \! \! \! \! \! \! \! \! \hat \psi_{\theta}^0(\zeta,\lambda)dd^c G_{\lambda,\theta}(z,\zeta)  \\
&\qquad \qquad +\lim_{R \to \infty} \int\limits_{|\zeta_1| =R}[\bar \partial G_{\lambda,\theta}(z,\zeta) \hat \psi_{\theta}^0(\zeta,\lambda)+ G_{\lambda,\theta}(z,\zeta)\partial \hat \psi_{\theta}^0(\zeta,\lambda)].
\end{align*}
The restriction of the last equation to the boundary $\partial X$ from outside yields \eqref{fred}.

\textit{ii)} To prove that \eqref{fred} is a Fredholm-type equation, for fixed $\lambda \in \compl \setminus E_{\theta}$, $|\lambda| \geq const(V, \{ a_j \}, \theta, \hat \sigma)$, we proceed as follows. Let $f(z) = \hat \psi_{\theta}(z,\lambda)-e^{\lambda(z_1+\theta z_2)}$ and $f^0(z) = \hat \psi_{\theta}^0(z,\lambda) -e^{\lambda(z_1+\theta z_2)}$; we can write equation \eqref{fred} as
\begin{equation} \label{newfred}
f+Tf=g,
\end{equation} 
where
\begin{align}
g(z)&=\frac{i}{2}\int_{\zeta \in \partial X}G_{\lambda,\theta}(z,\zeta)(\Lambda_{\hat \sigma}-\Lambda_0)e^{\lambda(\zeta_1 + \theta \zeta_2)} \\ \nonumber
&\qquad + i e^{\lambda (z_1 + \theta z_2)}\sum_{j=1}^g \hat C^0_{j,\theta}(\lambda)g_{\lambda,\theta}(z,\hat a_j), \\ \label{opt}
Tf(z)&= -\frac{i}{2}\int_{\zeta \in \partial X}G_{\lambda,\theta}(z,\zeta)(\Lambda_{\hat \sigma}-\Lambda_0)f(\zeta)\\ \nonumber
&\qquad -ie^{\lambda (z_1 + \theta z_2)} \sum_{j=1}^g \hat C^1_{j,\theta}(\lambda)g_{\lambda,\theta}(z,\hat a_j) \\\nonumber
&\qquad +  \lim_{\varepsilon \to 0} \frac i 2  \! \! \!  \! \! \! \!  \! \! \! \! \! \int\limits_{ \{\zeta \in V : |\zeta-z| \geq \varepsilon\} }\! \! \! \! \! \! \! \! \! \! f^0(\zeta)dd^c G_{\lambda,\theta}(z,\zeta) \\ \nonumber
&\qquad + \lim_{R \to \infty} \int\limits_{|\zeta_1| =R}[\bar \partial G_{\lambda,\theta}(z,\zeta) f^0(\zeta)+ G_{\lambda,\theta}(z,\zeta)\partial f^0(\zeta)],
\end{align}
where $\hat C^0_{j,\theta}+\hat C^1_{j,\theta} = \hat C_{j,\theta}$ ($C^0_{j,\theta}$ is obtained from \eqref{54} with $e^{\lambda(z_1+\theta z_2)}$ instead of $\hat \psi_{\theta}(z,\lambda)$, so it is independent from $f$).

We have now that equation \eqref{newfred} is a Fredholm-type integral equation for $f \in W^{1,2}(\partial X)$. Indeed $g \in W^{1,2}(\partial X)$ and $T$ is a compact operator: this follows from the compactness of $\Lambda_{\hat \sigma}-\Lambda_0$ for the first term in \eqref{opt}, from formulas \eqref{poisson1}, \eqref{poisson2} and \eqref{green73} for the third term, while the second and the fourth term are operators with finite-dimensional range.
\smallskip

The existence, for $\lambda \in \compl \setminus E_{\theta}$, $|\lambda| \geq const(V, \{ a_j \}, \theta, \hat \sigma)$, of a unique Faddeev-type function $\hat \psi_{\theta} (z,\lambda)$ imply the existence, for such $\lambda$, of a solution of (\ref{fred}) with residue data $i\hat C_{j,\theta} (\lambda)$, $j=1,\ldots,g$.

Let us prove the uniqueness, with $\lambda$ as above, of the solution of (\ref{fred}) in $W^{1,2}(\partial X)$. Suppose that $\hat \psi_{\theta} \in W^{1,2}(\partial X)$ solves (\ref{fred}), and consider $\hat \mu_{\theta} = e^{-\lambda(z_1 +\theta z_2)} \hat \psi_{\theta}$ as the Dirichlet data for 
$$\overline \partial (\partial + \lambda (dz_1 + \theta dz_2))\hat \mu_{\theta} = \frac i 2 Q \hat \mu_{\theta}$$
on X; thanks to this equation we can well define $\hat \mu_{\theta}$ on $\overline X$. We also define $\hat \mu_{\theta}$ on $V \setminus \overline X$ by (\ref{fred}). The function $\hat \mu_{\theta}$ then defined on $V$ belongs to $C(V \setminus \cup_{j=1}^g \{a_j \})$.

To show that $\hat \psi_{\theta} = e^{\lambda(z_1 +\theta z_2)} \hat \mu_{\theta}$ satisfies \eqref{faddeev}, \eqref{eqGut} globally, we can follow without modification the arguments of \cite[Prop. 5.1]{HN}, based on the Sohotsky-Plemelj jump formula.

The uniqueness of the solution of \eqref{fred} in $W^{1,2}(\partial X)$ with residue data $\{ \hat C_{j,\theta}\}$ now follows from the uniqueness for Faddeev-type functions for $\lambda \in \compl \setminus E_{\theta}$, $|\lambda| \geq const(V, \{ a_j \}, \theta, \hat \sigma)$.
\end{proof}

\section{Cauchy-type Formulas} \label{cauchy}
Following our reconstruction scheme, after recovering the boundary value of the Beltrami solution $F$, we obtain $F(\partial X) = \Gamma$.

Thus the remaining problem is reconstructing the interior points of a bordered Riemann surface $Y$ given the boundary $\Gamma$.

We will use the coordinates $z =(z_1, z_2) \in \compl^2$ and the projection $p: \compl^2 \to \compl$ on the first factor, $p(z) = z_1$. For $a \in \compl$ we define
$$N_{a} = \frac{1}{2 \pi i} \int_{\Gamma}\frac{ dz_1}{z_1 - a} \in \nat,$$
which counts the number of intersection points of the line $\{ z_1 = a \}$ with the surface $Y$ that we are going to reconstruct. Let us remark that, if we call $Y_1, \ldots, Y_s$ the bounded connected components of $\compl \setminus p (\Gamma)$, we have that $N_a$ is constant on every $Y_h$, $h=1,\ldots,s$.

We have the following proposition, the first part of which is a special case of a result by Harvey-Shiffman \cite{HSh}, while the second part goes back to Cauchy. 

\begin{prop} Let $\Gamma$ be a $C^1$-closed curve in $\compl^2$:
\begin{itemize}
\item[\it i)] if $Y_1, Y_2$ are two bordered Riemann surfaces in $\compl^2$ with the same boundary $\Gamma$, then $Y_1 = Y_2$;
\item[\it ii)] the interior points of the unique Riemann surface $Y$ whose boundary is $\Gamma$ can be explicitly found from the system of equation
\begin{equation} \label{system}
\frac{1}{2 \pi i} \int_{\Gamma} z_2^k(z_1) \frac{ dz_1}{z_1 - a} = \sum_{j=1}^{N_a} (z_2^{(j)})^k (a), \; \; k=1, \ldots, N_a.
\end{equation}
The points of the surface are the pairs $(a, z_2^{(j)}(a))$, for  $j=1,\ldots, N_a$, $a \in Y_h$, $h=1,\ldots, s$.
\end{itemize}
\end{prop}

By $i)$ we have that $F(X) = Y$; then, from the regularity assumptions on $X$ and $F$ we deduce that $Y$ is a Riemann surface with $C^1$ boundary.

\begin{proof} {\it ii)} Formulas \eqref{system} are true by residue theorem. Now, if $a \in Y_h$ for some $h$, since we know the Newton sums $\sum_{j=1}^{N_a} (z_2^{(j)})^k (a)$ for every $k$, we can find $(z_2^{(j)}) (a)$ by a well-known algebra result.
\end{proof}

\section{Reconstruction of $\sigma$} \label{reconstruction}
Thanks to the integral equation \eqref{fred} and formulas \eqref{recbel}, \eqref{identity}, we can find $\psi_{\theta} (w,\lambda)|_{\partial Y}$ from $\Lambda_{\hat \sigma}$, where $\psi_{\theta}$ is a Faddeev-type isotropic solution as in the proof of Theorem \ref{teofaddeev} and $Y$ is the reconstructed surface in section \ref{cauchy}.
By the remarks in section \ref{basic}, from $\Lambda_{\hat \sigma}$ and $F|_{\partial X}$ we can also find $\Lambda_{\sigma}$ on $\partial Y$.

Thus we have that $\Lambda_{\hat \sigma}$ determines $\Lambda_{\sigma}$ uniquely and $\psi_{\theta}(w,\lambda) |_{\partial Y}$, for $\lambda \in \compl \setminus E_{\theta}$, $|\lambda| \geq const(V, \{ a_j \}, \theta, \hat \sigma)$ and for $\theta \in \compl$. This will be sufficient to recover $\sigma$ on $Y$.

We define $\tilde \psi_{\theta} = \sqrt{\sigma} \psi_{\theta}$, so that by \eqref{isofaddeev} $dd^c \tilde \psi_{\theta} = q \tilde \psi_{\theta} + \sum_{j=1}^g C_{j,\theta}(\lambda) \delta (z,a_j)$, where $q = \frac{dd^c \sqrt{\sigma}}{\sqrt{\sigma}}$, and we have the following theorem:

\begin{theo}[Thm. 1.2B \cite{HN}] \label{recona}
The function $\sigma(w), \; w \in Y$, can be reconstructed from the Dirichlet-to-Neumann data
$$\tilde \psi_{\theta}|_{\partial Y} = \mu_{\theta}|_{\partial Y} e^{\lambda (z_1+\theta z_2)} \to \overline \partial \tilde \psi_{\theta}|_{\partial Y}$$
using an explicit formula. In particular, for the case $W= \{ z \in \compl^2 : P(z) =0 \}$, where $P$ is a polynomial of degree $N$, this formula is as follows. Let $\{ w_m \}$ be the points of $W$ where $(dz_1 + \theta dz_2) |_W (w_m) = 0$, $m=1, \ldots, M$. Then, for almost every $\theta$, the value $\frac{dd^c \sqrt{\sigma}}{\sqrt{\sigma} dd^c |z|^2}|_W (w_m)$ can be found from the following linear system:
\begin{align} \label{sysrec}
\tau(1+ o(1))& \frac{d^k}{d \tau^k}\left( \int_{z \in \partial Y} e_{i\tau, \theta} \overline \partial \mu_{\theta}(z,i \tau) \right)\\ \nonumber
&= \tau(1+ o(1)) \frac{d^k}{d \tau^k}\left( \int_{z \in Y} e_{i\tau, \theta} q \mu_{\theta}(z,i \tau) \right) \\ \nonumber
&= \left. \sum_{m=1}^M \frac{i \pi (1 + |\theta|^2)}{2}\frac{dd^c \sqrt{\sigma}}{\sqrt{\sigma} dd^c |z|^2}\right|_W (w_m) \\ \nonumber
&\times \frac{|\frac{\partial P}{\partial z_1}(w_m)|^3 \frac{d^k}{d \tau^k} \exp i \tau [(w_{m,1} + \theta w_{m,2}) + (\overline w_{m,1} + \overline \theta \overline w_{m,2})]}{|\frac{\partial^2 P}{\partial z_1^2}(\frac{\partial P}{\partial z_2})^2 -2 \frac{\partial^2 P}{\partial z_1 \partial z_2}(\frac{\partial P}{\partial z_2})(\frac{\partial P}{\partial z_1}) + \frac{\partial^2 P}{\partial z_2^2} (\frac{\partial P}{\partial z_1})^2|(w_m)}
\end{align}
where $m, k = 1,\ldots, M ; \; M =N(N-1), \; \tau \in \real, \; \tau \to \infty$ such that $|\tau|^g |\Delta_{\theta} (i \tau)| \geq \varepsilon > 0$, with $\varepsilon$ small enough. The determinant of system \eqref{sysrec} is proportional to the Vandermonde determinant of the points $\{ (w_{m,1} + \theta w_{m,2}) + (\overline w_{m,1} + \overline \theta \overline w_{m,2}) \}$.
\end{theo}

The proof of this theorem is given in \cite{HN}, under the condition that $Sing \, Y = \emptyset$; nevertheless, the proof is still valid if $Y$ contains self-inter\-sec\-tion-type singularities.

To apply Theorem \ref{recona}, since $\tilde \psi_{\theta}|_{\partial Y} = \psi_{\theta}|_{\partial Y}$ we only need to show that the integral 
$$\int_{\partial Y} e_{\lambda, \theta} \overline \partial \mu_{\theta}(z,\lambda) = \int_{\partial Y} e^{- \overline \lambda (\overline z_1 + \overline \theta \overline z_2)} \overline \partial \psi_{\theta} (z, \lambda), \; \; \lambda \in \compl$$
can be expressed in terms of $\Lambda_{\sigma}$. This is a consequence of the following lemma.

\begin{lem}
For every $\phi \in C^1(\partial Y)$ and every $\psi \in C^1(Y)$ solution of $d \sigma d^c \psi=  ( dd^c - M)\psi = 0$ in $Y$, we have

\begin{equation}
\int_{\partial Y} \phi(\Lambda_{\sigma}-\Lambda_0)\psi = 2i \int_{\partial Y} \phi (\overline \partial \psi - \overline \partial \psi_0),
\end{equation}
where $dd^c \psi_0 = 0$ in $Y$ and $\psi_0|_{\partial Y}=\psi|_{\partial Y}$. 
\end{lem}

\begin{proof}
Let $a \in C^1(Y)$ such that $a|_{\partial Y}=\phi$. From the definition of the Dirichlet-to-Neumann operator and from Stokes' theorem, one has
$$\int_{\partial Y} \phi (\Lambda_{\sigma}-\Lambda_0)\psi = \int_{Y} (d a \wedge d^c (\psi - \psi_0) + a M\psi),$$
and, with the identity $dd^c = 2i \partial \overline \partial$, Stokes' theorem gives
$$2i \int_{\partial Y} \phi(\overline \partial \psi - \overline \partial \psi_0) = 2i\int_{Y}  \partial a  \wedge (\overline \partial \psi - \overline \partial \psi_0) + \int_Y a M \psi \; dx dy.$$
Expressing the first integrand on the right in coordinate form we get $$\partial a \wedge \overline \partial (\psi-\psi_0) = d a \wedge \overline \partial (\psi-\psi_0)= \frac{1}{2i}d a \wedge d^c (\psi - \psi_0) + \frac{1}{2}da \wedge d(\psi-\psi_0).$$
Again by Stokes' thorem we have
$$\int_{Y}da \wedge d(\psi - \psi_0) = - \int_{\partial Y} (\psi-\psi_0)da=0$$
because $\psi|_{\partial Y} = \psi_0|_{\partial Y}$. The proof follows.
\end{proof}

If we put $\phi = e^{- \overline \lambda (\overline z_1 + \overline \theta \overline z_2)}|_{\partial Y}$, we find that
\begin{align*}
\frac{1}{2 i}\int_{\partial Y} e^{- \overline \lambda (\overline z_1 + \overline \theta \overline z_2)} (\Lambda_{\sigma}-\Lambda_0)\psi_{\theta} &= \int_{\partial Y} e^{- \overline \lambda (\overline z_1 + \overline \theta \overline z_2)} (\overline \partial \psi_{\theta} - \overline \partial \psi_0) \\
&= \int_{\partial Y} e^{- \overline \lambda (\overline z_1 + \overline \theta \overline z_2)} \overline \partial \psi_{\theta},
\end{align*}

because $ \partial e^{- \overline \lambda (\overline z_1 + \overline \theta \overline z_2)} = 0$ and $\partial \overline \partial \psi_0 = 0$ on $Y$.

\section{Proof of Theorem \ref{maintheo}}

We now put together all the results of this paper to prove the main theorem and his corollary.

\begin{pfmain}
We start finding a complex structure on $X$. This is done by a standard construction, as suggested in the introduction. We consider the local form of the Euclidean metric of $\real^3$ restricted to $X$:
$$ds^2 = E dx^2 + 2F dx dy + G dy^2$$
where $x, y$ are oriented coordinates. Let $z = x+iy$, and define 
$$\mu (z) = \frac{\frac 1 2  (E-G) +iF}{\frac 1 2 (E+G) + \sqrt{EG-F^2}}.$$
Then the local homeomorphic solutions of the Beltrami equation $\frac{\partial w}{\partial \overline z}= \mu (z) \frac{\partial w}{\partial z}$ form a holomorphic atlas on $X$, which then becomes a Riemann surface.

We now embed $X$ in $\proj^3$ -- as explained in section \ref{prelim} -- as an open set of a nonsingular affine algebraic curve $V$. By proposition \ref{propBelAn}, there exists a unique $C^1$-quasiconformal diffeomorphism $F:V \to W$ with special asymptotic conditions such that $F_{\ast} \hat \sigma = \sigma$ is isotropic on $W$.

Starting from $\Lambda_{\hat \sigma}$ we first recover $\hat \psi_{\theta}(z, \lambda) |_{\partial X}$ by integral equation \eqref{fred}, and then $F|_{\partial X}$ by formula \eqref{recbel}.

Successively, from the knowledge of $F(\partial X) = \partial Y$, we reconstruct $Y$ using the formulas \eqref{system}. Finally we can reconstruct $\sigma$ on $Y \setminus Sing(Y)$ with the help of Theorem \ref{recona} and the remarks in section \ref{reconstruction}.

If $\tilde Y$, $\tilde \sigma$ and $\tilde F : X \to \tilde Y$ are as in the statement of the theorem, then $\Psi = \tilde F \circ F^{-1} : Y \to \tilde Y$ is weakly holomorphic because $\tilde F$ satisfies the same Beltrami equation as $F$. By properties of $F$ we have that $\Psi : Y \setminus Sing(Y) \to \tilde Y \setminus \Psi( Sing(Y))$ is a biholomorphism which can be uniquely extended to a biholomorphism $\Psi' : Y' \to \tilde Y'$, where $Y'$ and $\tilde Y'$ are normalizations of $Y$ and $\tilde Y$ respectively. Properties of $F$ allow us also to extend smoothly $\sigma$ and $\tilde \sigma$ on $Y'$ and $\tilde Y'$ as $\sigma'$ and $\tilde \sigma'$ respectively. Finally we obtain $\Psi'_{\ast} \sigma' = \tilde \sigma'$, which ends the proof.
\end{pfmain}

\begin{pfcoro}
If we require that $F$ has the special asymptotics as in proposition \ref{propBelAn}, then the whole construction in Theorem \ref{maintheo} is unique.

Taking account of this, if $\Lambda_{\sigma_1} = \Lambda_{\sigma_2}$ we have $F_1|_{\partial X} = F_2|_{\partial X}$, where $F_1, F_2$ are the special quasiconformal solutions given by proposition \ref{propBelAn} associated to $\sigma_1$ and $\sigma_2$ respectively. Thus we also obtain, from $F_1(\partial X) = F_2(\partial X)=\partial Y$ and the formulas \eqref{system}, that $F_1(X) = F_2(X)=Y$. Let $G :Y' \to Y$ be a normalization of $Y$, and $F'_j = G^{-1} \circ F_j: X \setminus F_j^{-1}(Sing(Y)) \to Y' \setminus G^{-1} (Sing(Y))$, for $j= 1,2$. Then, by construction, $F'_j$ can be extended as a global diffeomorphism between $X$ and $Y'$, for $j = 1,2$. Now, if we define the smooth isotropic conductivities on $Y'$ as $\sigma'_j = (F'_j)_{\ast} \sigma_j$, $j=1,2$, we find $\Lambda_{\sigma'_1} = \Lambda_{\sigma'_2}$, and the boundary values of the respective Faddeev-type anisotropic (resp. isotropic) solutions coincide on $\partial X$ (resp. $\partial Y'$). Consequently $\sigma'_1 = \sigma'_2$ on $Y'$.

We finally define $\Phi = {F'}_2^{-1} \circ F'_1 : \overline X \to \overline X$ which satisfies $\Phi|_{\partial X} = \mathrm{Id}$ and $\Phi_{\ast} \sigma_1 = \sigma_2$.
\end{pfcoro}

\pagebreak

\end{document}